\documentclass{amsart}
\usepackage{amssymb, amsmath, color}
\usepackage[latin1]{inputenc}
\usepackage{graphicx}
\usepackage{color}

\newtheorem{theorem}{Theorem}
\newtheorem{lemma}[theorem]{Lemma}
\newtheorem{remark}[theorem]{Remark}
\newtheorem{definition}[theorem]{Definition}
\newtheorem{corollary}[theorem]{Corollary}

\begin{document}

\title[Relative position between a paraboloid and an ellipsoid]
{Classification of the relative positions between an ellipsoid and an elliptic paraboloid}
\author[Brozos-V\'{a}zquez \, Pereira-Sáez \, Souto-Salorio \,  Tarrío-Tobar]{M. Brozos-V\'{a}zquez \, M.J. Pereira-Sáez \, M.J. Souto-Salorio \, Ana D. Tarrío-Tobar}
\address{MBV: Universidade da Coru\~na, Differential Geometry and its Applications Research Group,  Escola Polit\'ecnica Superior,  Spain}
\email{miguel.brozos.vazquez@udc.gal}
\address{MJPS: Universidade da Coru\~na, Departamento de Econom\' ia, Facultade de Econom\' ia e Empresa, Spain}
\email{maria.jose.pereira@udc.es}
\address{MJSS: Universidade da Coru\~na, Departamento de Computaci\'on, Facultade de Inform\'atica, Spain}
\email{ maria.souto.salorio@udc.es}
\address{ADTT: Universidade da Coru\~na, Departamento de Matem\'{a}ticas, Escola T\'ecnica de Arquitectura, Spain}
\email{madorana@udc.es}
\thanks{Supported by Projects ED431F 2017/03, MTM2016-75897-P and MTM2016-78647-P (AEI/FEDER, UE)}

\begin{abstract}
We classify all the relative positions between an ellipsoid and an elliptic paraboloid when the ellipsoid is small in comparison with the paraboloid ({\it small} meaning that the ellipsoid cannot be tangent to the paraboloid at two points simultaneously). This provides an easy way to detect contact between the two surfaces by a direct analysis of the coefficients of a fourth degree polynomial.
\end{abstract}

\maketitle

\section{Introduction}\label{sect:introduction}
As the simplest curved surfaces, quadrics have been extensively used  in CAD/CAM and industrial design. This is justified by the fact that they can be manipulated through a simple algebraic expression, which makes their geometry more tractable for computational purposes. Another advantage is that quadric surfaces can be used to approximate other more complicated surfaces locally up to order two or to build a variety of shapes piecewise \cite{levin76,patricalakis-maekawa,wang2002,yan-wang-yang}.


The problem of contact detection is essential in fields such as robotics, computer graphics, computer animation, CAD/CAM, etc. There exists an extensive literature about contact detection between quadric surfaces (see \cite{Levin79} for a classical reference and references therein such as \cite{Wang2009}), but one of the main tools relies on the use of a characteristic polynomial associated to the pencil of the quadrics to treat the collision detection problem. Indeed that was done in previous works considering conic curves (see \cite{alberich2017,Choi2006,etayo2006,Liu2004}) and other quadric surfaces, especially ellipsoids, in Euclidean or Projective spaces  (see \cite{X2011,wank-wang-kim,Wang03}).

Although the literature dealing with the relative positions of two ellipsoids is large (see \cite{Choi2003,Wang2004,wank-wang-kim} among others), there is a lack of information on the same problem associated to other quadrics that would fit different geometric features better than an ellipsoid when considered in practical contexts. With this motivation, we address the problem of finding the relative position between an ellipsoid and an elliptic paraboloid. 

A direct approach to the problem of looking for the intersection between an ellipsoid and an elliptic paraboloid involves to solve a system of two polynomial equations of degree two. In order to avoid this  task, the proposed method provides a quick answer in terms of the coefficients of a polynomial of degree four. Additionally, the resulting algorithms do not only discern whether there is contact or not but also give information on the relative position between the quadrics. Thus, this work sets the theoretical foundations of a method to detect the relative position between the two  quadrics for the follow-up applications in specific practical contexts.

The classification of the relative positions leads to an extra natural assumption related to the size of the ellipsoid in comparison with the size of the paraboloid. This is given in Section~\ref{section:classification} and further explained in the Appendix (see Section~\ref{section:appendix-smallness}). The paper is organized as follows. The characterization of relative positions in terms of the roots of the characteristic polynomial associated to the pencil of the quadrics is given in  Section~\ref{section:classification}, where possible applications are also treated. In Section~\ref{section:characteristic-polynomial} we analyze the characteristic polynomial and reduce the problem to study a paraboloid and a sphere. The particular case where the center of this sphere is in the $OZ$-axis is studied in Section~\ref{section:oz-axis}, then extended to the whole space in Section~\ref{section:relative positions}, to prove the main result in Section~\ref{section:theproofofth}. Conclusions are summarized in Section~\ref{section:conclusions} and, finally, some complementary material is provided in the Appendix to further enlighten more technical aspects of the developed approach.

\section{Classification of the relative positions}\label{section:classification}

\subsection{Quadrics and characteristic polynomial.}\label{subsect:quadrics}

We consider a general elliptic paraboloid $\mathcal{P}$ with associated matrix $P$ and a general ellipsoid $\mathcal{E}$ with associated matrix $E$, in such a way that in homogeneous coordinates $X=(x,y,z,1)^t$ the corresponding equations are, respectively,
$X^t P X=0$ and $X^t E X=0$.

We consider the pencil of the quadrics  $\lambda P+E$ in order to define the \emph{ characteristic polynomial} of $\mathcal{P}$ and $\mathcal{E}$ as
\begin{equation}\label{eq:characteristicpolynomial}
f(\lambda)=\operatorname{det}(\lambda P+E).
\end{equation}
We will refer to the roots of the polynomial $f$ as the \emph{characteristic roots} of $\mathcal{P}$ and $\mathcal{E}$.
\subsection{The ``smallness'' condition}
All along this work we will assume that the ellipsoid is {\it small} in comparison with the elliptic paraboloid. More precisely, we will assume the following: 

\smallskip

{\it ``Smallness'' hypothesis: the ellipsoid and the elliptic paraboloid cannot be tangent at two points simultaneously.}

\smallskip

This hypothesis has a very specific geometric meaning and is naturally motivated by the classification of relative position between quadrics. This condition is equivalent to the fact that the ellipsoid and the paraboloid intersect in just one point or a curve with only one connected component. Thus, other more complicated relative positions that include multiple tangent points or intersections in curves with two connected components are excluded by this hypothesis. 
Moreover, the ``smallness'' hypothesis can be phrased in terms of principal curvatures of the surfaces: the smallest principal curvature in the ellipsoid is greater than the largest principal curvature of the elliptic paraboloid (which is attained at the vertex point). We refer to the Appendix for details about this geometric approach and a characterization in terms of the parameters of the quadrics.

The avoidance of the ``smallness'' hypothesis leads to complicated positions which, in principle, are not so interesting if one only intends to detect contact between the surfaces (see, for example, \cite{Jia-Wang-Choi-Mourrain-Tu1,Levin79,Wang-Mourrain-Wang-2009,Wang-Goldman-Tu2003} for information on the intersection curves between quadric surfaces).

\subsection{Relative positions between the surfaces}
To analyze the relative positions between the two quadrics we establish the following definition.

\begin{definition} We say that:
	\begin{enumerate}
		\item 
		$\mathcal{P}$ and $\mathcal{E}$ are in \emph{contact} if there exists $X$ such that $X^tPX=X^tEX=0$;
		\item $\mathcal{P}$ and $\mathcal{E}$ are \emph{tangent} at a point $X$ if  $X^tPX=X^tEX=0$ and $Y^tPX=0$ if and only if $Y^tEX=0$, this is, $\mathcal{P}$ and $\mathcal{E}$ are in contact at $X$ and have the same tangent plane at $X$;
		\item a point $X$ is \emph{interior} to $\mathcal{P}$ if $X^tPX<0$ and \emph{exterior} to  $\mathcal{P}$ if $X^tPX>0$. By extension, we also say that $\mathcal{E}$ is interior (or exterior) to $\mathcal{P}$ if every point in $\mathcal{E}$ is interior (or exterior) to $\mathcal{P}$. 
	\end{enumerate}
\end{definition}

\begin{theorem}\label{th:th1}
Let $\mathcal{E}$ be an ellipsoid and $\mathcal{P}$  an elliptic paraboloid satisfying the ``smallness'' condition. The possible relative positions between $\mathcal{E}$ and $\mathcal{P}$ are those given in Table~\ref{table:main} below and they are characterized by the corresponding configuration of the characteristic roots.

\begin{table}
	\begin{tabular}{|c|c|}	
		\hline
		\multicolumn{2}{|c|}{\bf \phantom{123} Relative positions between an elliptic paraboloid $\mathcal{P}$ and a ``small'' ellipsoid $\mathcal{E}$\phantom{123}}	\\
		\multicolumn{2}{|c|}{\begin{minipage}{12cm} \small There are always two roots ($\lambda_1$ and $\lambda_2$) that satisfy:
	$\lambda_1\leq\lambda_2<0$, the configuration of the other two roots ($\lambda_3$ and $\lambda_4$) determines the relative position of the ellipsoid and the elliptic paraboloid.\end{minipage}} \\	
		\hline
		\multicolumn{2}{|c|}{
		$\mathcal{E}$ is interior to $\mathcal{P}$ (Type I) or $\mathcal{E}$ is tangent to $\mathcal{P}$ from inside (Type TI)} \\
		\multicolumn{2}{|c|}{
		\includegraphics[scale=0.4]{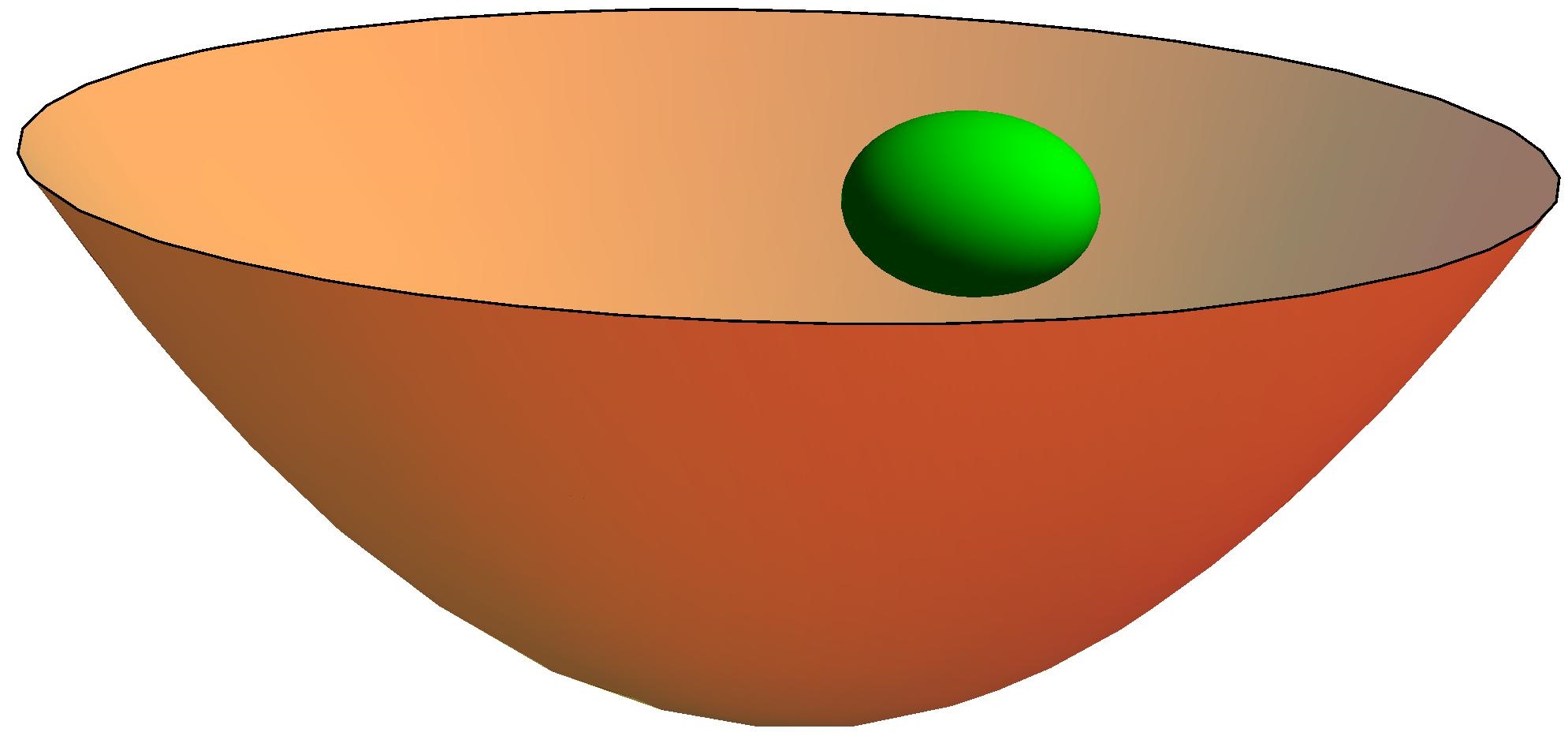}\qquad\qquad
			\includegraphics[scale=0.5]{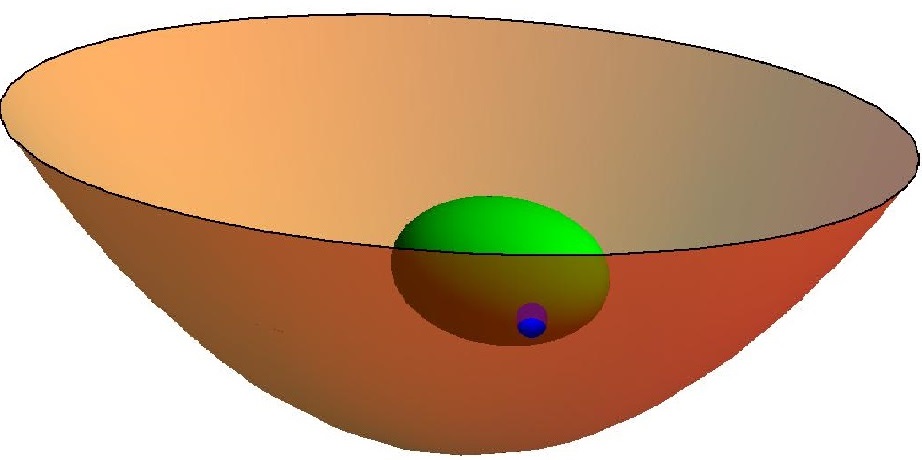}}\\
		\multicolumn{2}{|c|}{ 
			{\bf Roots:} }\\
		\multicolumn{2}{|c|}{$\lambda_3\leq \lambda_4<0$}\\
		\hline
			\multicolumn{2}{|c|}{Contact between $\mathcal{E}$ and $\mathcal{P}$ (Type C)}\\
			\multicolumn{2}{|c|}{	\includegraphics[scale=0.5]{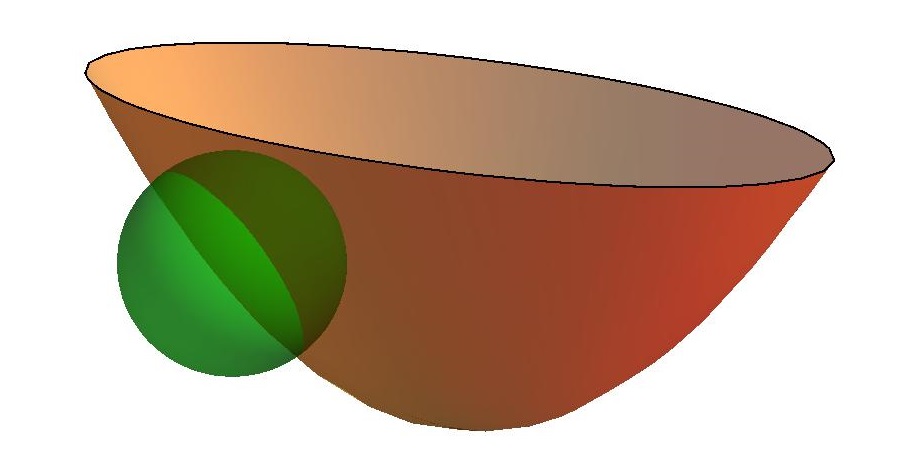}}\\
				\multicolumn{2}{|c|}{\bf Roots:}\\
					\multicolumn{2}{|c|}{$\lambda_3=\bar \lambda_4\notin\mathbb{R}$}\\
		\hline
			$\mathcal{E}$ is tangent to $\mathcal{P}$ from outside (Type TE)& $\mathcal{E}$ is exterior to $\mathcal{P}$ (Type E)\\
		\includegraphics[scale=0.5]{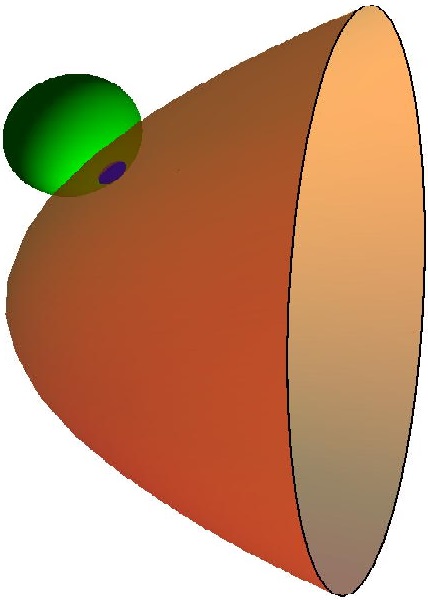}
		&	\includegraphics[scale=0.4]{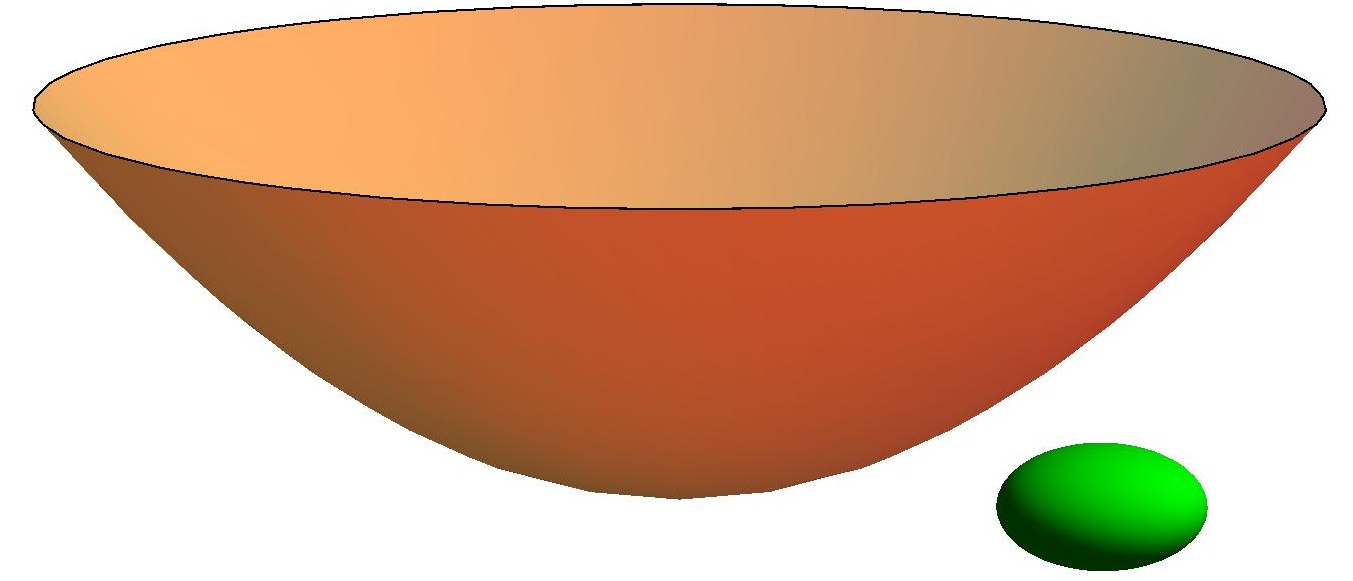}\\
		{\bf Roots:}& {\bf Roots:}\\
		$0<\lambda_3= \lambda_4$ & $0<\lambda_3<\lambda_4$\\
		\hline
	\end{tabular}
	\caption{Characterization of the relative positions between  a small ellipsoid and an elliptic paraboloid in terms of the characteristic roots.}
	\label{table:main}
\end{table}
\end{theorem}

\begin{remark}\label{remark:IeTI}\rm
	Types I and TI are both detected by four negative roots. The problem of distinguishing these two relative positions comes from the fact that if there are exactly two roots that can be double roots in the Type I position (called them $\bar \lambda$ and $\tilde \lambda$). These two special roots are determined in terms of the parameters of the quadrics (see Theorem~\ref{th:tangency}).  Thus, one can distinguish between the two relative positions once it is known that  $\bar \lambda$ and $\tilde \lambda$ are not characteristic roots. Hence, one has that:
	\begin{enumerate}
		\item $4$ different real negative roots imply Type I;
		\item $3$ different real negative roots, one of which has multiplicity two and is different from $\bar \lambda$ and $\tilde \lambda$, imply Type TI.
	\end{enumerate}
\end{remark}

Based on Theorem~\ref{th:th1}, one can detect the relative position of $\mathcal{E}$ and $\mathcal{P}$ by a direct analysis of the coefficients of the characteristic polynomial $f$ as follows.
The first step will be to compute the characteristic polynomial $f$ and transform it in a monic polynomial $p$:
\[
p(\lambda)=\lambda^4+ c_3 \lambda_3+c_2 \lambda_2+c_1 \lambda+c_0.
\]
The discriminant of the fourth degree polynomial $p(\lambda)$ is given by the expression (see, for example, \cite{Irving}):
\[
\begin{array}{l} 
\Delta =    256c_0^3 - 192c_3c_1c_0^2-128c_2^2c_0^2 + 144c_2c_1^2c_0 - 27 c_1^4 
+ 144c_3^2c_2c_0^2  -  6c_3^2c_1^2 c_0 -  80c_3c_2^2c_1c_0   \\
\noalign{\smallskip}
\phantom{\Delta =} +   18c_3c_2c_1^3 + 16c_2^4c_0 -4c_2^3c_1^2 - 27 c_3^4c_0^2     +18c_3^3c_2c_1c_0 - 4c_3^3c_1^3-4c_3^2c_2^3c_0 + c_3^2c_2^2c_1^2.
\end{array}
\]

In view of Theorem~\ref{th:th1}, there are neither two pairs of complex conjugate roots nor a double root and a pair of complex conjugate roots. Therefore, the non-tangent contact position is characterized by two complex conjugate roots (see Table~\ref{table:main}), so it is also characterized by $\Delta<0$ (see, for example, \cite{Jacobson}). The case $\Delta=0$ characterizes multiple roots, which are related with tangent contact. More specifically, double roots indicate tangent contact except in two particular cases which are directly determined by the parameters of the paraboloid (for a paraboloid in standard form as in \eqref{eq:paraboloid} the exceptions are  $-a^2$ or $-b^2$, see Section~\ref{subsect:tangency-multiplicity} for details). $\Delta>0$ implies that there are $4$ different real roots, which results in non-contact relative positions. In the later case, the Descartes' rule of signs allows to distinguish the interior and exterior cases as in Table~\ref{table:coefficients}. This is summarized in the following result.

\begin{corollary}
Let $\mathcal{E}$ be an ellipsoid and $\mathcal{P}$  an elliptic paraboloid satisfying the ``smallness'' condition. The relative positions between $\mathcal{E}$ and $\mathcal{P}$ are detected in terms of the coefficients of $p$ as shown in Table~\ref{table:coefficients}.

\begin{table}
	\begin{tabular}{|c|c|}
		\hline
		\multicolumn{2}{|c|}{Conditions to classify the relative positions between  $\mathcal{E}$ and $\mathcal{P}$.}\\
		\hline
		Type&  Conditions on the coefficients of $p(\lambda)$\\
		\hline\hline
		I or TI &  \begin{tabular}{c}
		$\Delta \geq 0$ and $c_i\geq 0$ for all $i=1,2,3$
		\end{tabular}		 \\
		\hline
		C & $\Delta<0$\\
		\hline
		TE& \begin{tabular}{c}
			$\Delta=0$ and $c_i<0$ for some $i=1,2,3$
		\end{tabular}\\
		\hline
		E& \begin{tabular}{c}
			$\Delta>0$ and $c_i<0$ for some $i=1,2,3$
		\end{tabular}\\
		\hline
	\end{tabular}
\caption{Characterization of the relative positions between  a small ellipsoid and an elliptic paraboloid in terms of the coefficients of the characteristic polynomial.}
\label{table:coefficients}
\end{table}
\end{corollary}


Possible algorithms to detect the relative position follow from Table~\ref{table:coefficients} and the analysis in the following sections. The data needed to analyze the positions is small enough for the implementation of efficient algorithms to be used in real-time continuous moving frameworks. For example, an algorithm to detect the relative position of an ellipsoid that moves continuously from the exterior of the elliptic paraboloid can be based only in the discriminant computation as follows:

\begin{itemize}
	\item If $\Delta<0$ then non-tangent contact (Type C), 
	\begin{itemize}
		\item[$\bullet$] else if $\Delta>0$ then exterior (Type E),\\
			\mbox{}\qquad	else tangent contact (Type TE).
	\end{itemize}
\end{itemize}
The previous algorithm just intends to illustrate the possibility of reasonable applications from Table~\ref{table:coefficients}. 
A deeper analysis in the development of algorithms should take care of the efficiency in the implementation and can follow the line of techniques already used to detect the positional relationship between conics or ellipsoids, we refer to \cite{alberich2017,etayo2006} and references therein.

\section{The characteristic polynomial}\label{section:characteristic-polynomial}

The set of roots of the characteristic polynomial given in \eqref{eq:characteristicpolynomial} is invariant under the action of the affine group. Indeed, for any nonsingular transformation $T$ one has that $\operatorname{det}(\lambda T^{t} PT+T^{t}ET)=\operatorname{det}(T)^2 \operatorname{det}(\lambda P+E)$. Therefore, in order to simplify the analysis of the relative positions we can always apply the composition of homotheties and rigid moves to the ellipsoid $\mathcal{E}$ and the paraboloid $\mathcal{P}$ as follows. First, by applying appropriate homotheties to the ellipsoid, it can be transformed into a sphere
 $\mathcal{S}$ of radius $r>0$ and center at $(x_c,y_c,z_c)$ (see, for example, \cite{affine-transf-book}):
\begin{equation}\label{eq:sphere}
\mathcal{S}:(x-x_c)^2+(y-y_c)^2+(z-z_c)^2=r^2\,.
\end{equation}
Thus, in homogeneous coordinates $X=(x,y,z,1)^t$, $\mathcal{S}$ is given by the equation $X^t S X=0$ where
\[
S=\left(   \begin{array}{rrrr}
1&0&0&-x_c \\
0&1&0&-y_c \\
0&0&1&-z_c \\
-x_c&-y_c&-z_c&-r^2+x_c^2+y_c^2+z_c^2
\end{array} \right ).
\]
Note that it would be possible to transform the sphere into one of radius $1$ too, but it does not carry a strong simplification and we prefer to work with a general radius $r$ to emphasize the role played by the radius along the subsequent arguments.
After converting the ellipsoid $\mathcal{E}$ into a sphere $\mathcal{S}$ by a transformation that also affects the paraboloid $\mathcal{P}$, we can apply an appropriate rotation and a translation  to the new elliptic paraboloid so that it is given in standard form (see \cite{affine-transf-book}): 
\begin{equation}\label{eq:paraboloid}
\mathcal{P}:\frac{x^2}{a^2}+\frac{y^2}{b^2}-z=0\,  \text{ for }  0<a\leq b.
\end{equation}
The $4\times 4$ matrix associated to $\mathcal{P}$ has the form
\[
P=\left(   \begin{array}{cccc}
a^{-2} &0&0&0 \\
0&b^{-2} &0&0\\
0&0&0&-1/2 \\
0&0&-1/2 &0
\end{array} \right ).
\]

We emphasize that this process does not carry a loss of generality since the relative position between the two quadrics and the characteristic roots remain unchanged. Therefore, in what follows, we work with a sphere $\mathcal{S}$  of radius $r$ and center $(x_c,y_c,z_c)$ as in \eqref{eq:sphere} and an elliptic paraboloid $\mathcal{P}$ given in standard form as in \eqref{eq:paraboloid}.

For $\mathcal{S}$  and $\mathcal{P}$ we compute the characteristic polynomial explicitly to obtain:
\begin{equation}\label{eq:characteristic-polynomial}
\begin{array}{rcl}
f(\lambda)&=&-\frac{1}{4a^{2}b^{2}}\left\{\lambda^{4}+(4z_c+a^{2}+b^{2}) \lambda^{3} \right.
 \\
\noalign{\medskip}
&&\left.+   (4z_c(a^{2}+ b^{2}) - 4(x_c^{2}+ y_c^{2}- r^2) +a^2b^2)\lambda^{2}\right.  \\
\noalign{\medskip}
&&\left. +  4  
( z_c a^2 b^2 -y_c^2 a^{2} -x_c^2 b^{2} +r^2(a^{2} +b^{2} ))\lambda\right\}  -r^2.
\end{array}
\end{equation} 
 The following is a general remark which will be crucial in the subsequent analysis. 
\begin{lemma}\label{lemma:tecnical-results-1}
	The characteristic roots of $f$ satisfy:
	\begin{enumerate}
		\item $0$ is not a root.
		\item The product of all roots is $4a^2b^2r^2>0$.
	\end{enumerate}
\end{lemma}
\proof
Substituting in the expression of the characteristic polynomial, we see that $f(0)=-r^2\neq 0$, so $0$ is not a root of $f$. From expression \eqref{eq:characteristic-polynomial}, transform $f$ to a monic polynomial multiplying by $-4 a^2b^2$ to see that the independent coefficient, which equals the product of the roots, is $4a^2b^2r^2>0$.
\qed

Since we are going to work with a sphere $\mathcal{S}$  of radius $r$, the ``smallness'' condition becomes quite tractable, as it can be expressed in terms of the parameters of the paraboloid and the radius of the sphere.
	The condition for the circumference not to be tangent at two points of the parabola (see Figure~1) is that the curvature of the circle is greater than or equal to the curvature of the parabola at any point. Since the curvature of the circle at any point is $\frac{1}{r}$ and the maximum curvature of the parabola is $\frac{2}{a^2}$, this condition reads:
	\[
	\frac{1}{r}\geq \frac{2}{a^2} \text{ or, equivalently, } 2 r\leq a^2.
	\]

 \begin{figure}[h]
	\includegraphics[width=3cm]{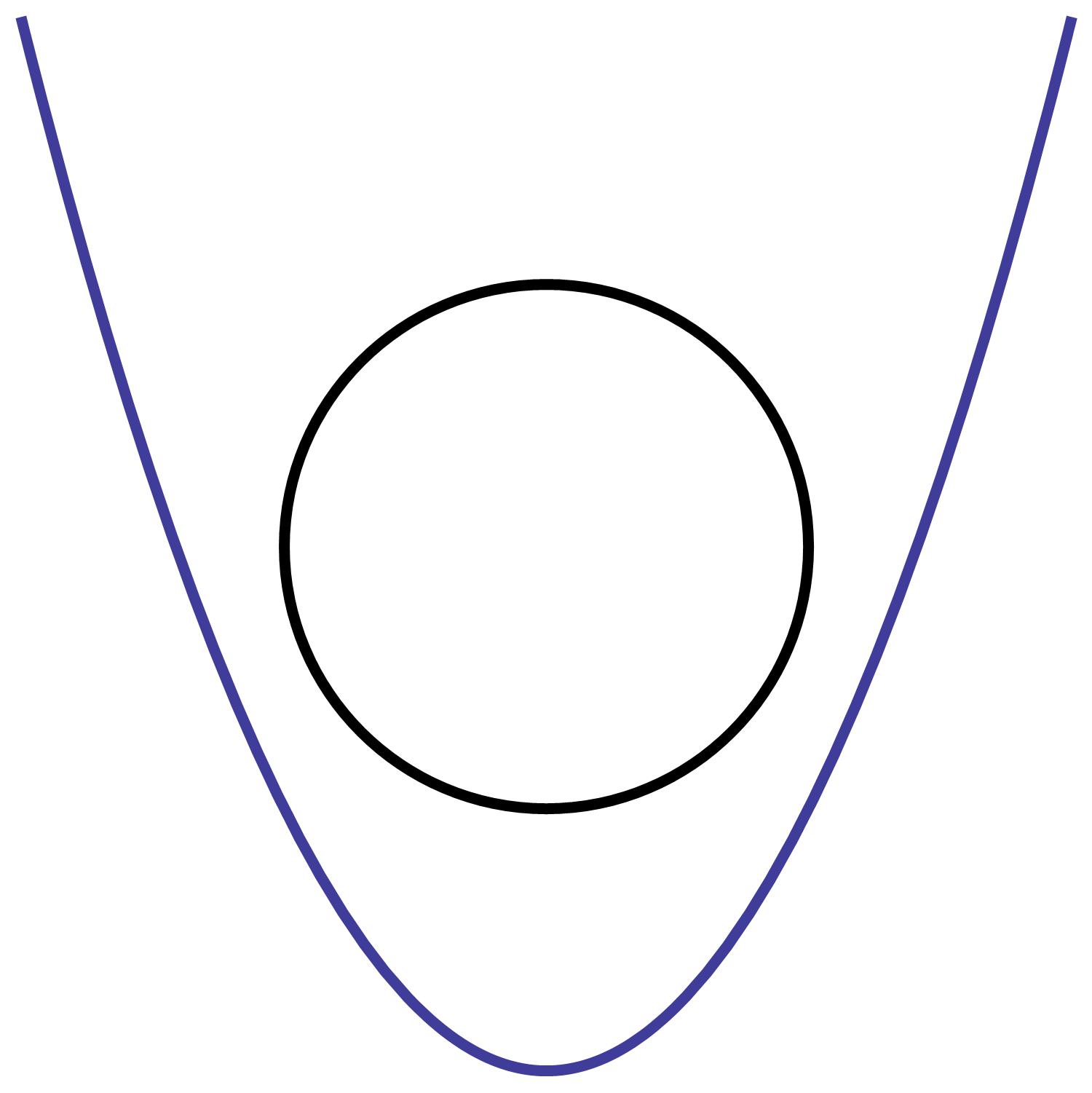}
	
	\caption{The maximum curvature of the parabola is attained at the vertex point. The condition for the circle not to be tangent at two points is that the curvature of the parabola at the vertex is less than or equal to the curvature of the circle.}\label{figure:parabola1}
\end{figure}



\section{Relative positions when the center of $\mathcal{S}$ is located at the $OZ$-axis.}\label{section:oz-axis}
As a first step in our approach to classify the relative positions of $\mathcal{S}$ and $\mathcal{P}$ in relation with the roots of $f$, we are going to consider the particular case in which the center of $\mathcal{S}$  is located in the $OZ$-axis. This location can be detected in terms of the roots of $f$ as follows.  Along this section we analyze the case $a\neq b$, as the case $a=b$ will be obtained as a consequence of this one. We adopt the notation established in Section~\ref{section:characteristic-polynomial} and assume the ``smallness'' hypothesis.

\begin{lemma}\label{lemma:-a2y-b2root}
	Assume $a\neq b$. The center of $\mathcal{S}$   is of the form $(0,0,z_c)$ if and only if 
 $-a^2$ and $-b^2$ are roots of $f$.
\end{lemma}
\proof
For $a\neq b$, substitute using expression \eqref{eq:characteristic-polynomial}: $f(-a^2) =x_c^2 (-b^2+ a^2)/b^2$ and $f(-b^2) =-y_c^2 (-b^2+ a^2)/a^2$ and the result follows. 
\qed

In virtue of Lemma~\ref{lemma:-a2y-b2root}, the center of $\mathcal{S}$ is at the $OZ$-axis implies that $\lambda_1=-a^2$ and $\lambda_2=-b^2$ are characteristic roots. In this case, $f(\lambda)=-4(a^{-2}\lambda+1)(b^{-2}\lambda+1)h(\lambda)$ with $h(\lambda)=\lambda^2+4z_c\lambda+4r^2$. Then, denote by $\lambda_3\leq \lambda_4$ the other two roots and observe that
\begin{equation}\label{eq2roots}
\lambda_3=-2(z_c+\sqrt{z_c^2-r^2}) \text{ and } \lambda_4=-2(z_c-\sqrt{z_c^2-r^2}).
\end{equation}

Note that for $i=3,4$
$$\frac{1}{2}\frac{d\lambda_i}{dz_c}=-1+(-1)^i\frac{z_c}{\sqrt{z_c^2-r^2}}.$$
Since $|z_c/\sqrt{z_c^2-r^2}|>1$, the value of $\lambda_3$ decreases and $\lambda_4$ increases (both strictly) as the center of $\mathcal{S}$ ascends through the $OZ$-axis for $z_c>r$. Whereas  $\lambda_3$ increases and $\lambda_4$ decreases as the center of $\mathcal{S}$  ascends through the $OZ$-axis when $z_c<-r$.

The relative positions between $\mathcal{S}$ and $\mathcal{P}$ are summarized in the following lemma.

\begin{lemma}\label{centroejez}
Assume $a\neq b$. Let $\mathcal{S}$ be a sphere centered at $(0,0,z_c)$ and $\mathcal{P}$ a standard elliptic paraboloid, then $\lambda_1=-a^2$, $\lambda_2=-b^2$ and all its relative positions can be characterized as follows:
	\begin{enumerate}
		\item[(i)] $\mathcal{S}$ is interior to $\mathcal{P}$ (Type I) if and only if $\lambda_3<\lambda_4<0$ 
		($\lambda_4\neq -a^2, -b^2$).
		\item[(ii)]  $\mathcal{S}$ is tangent from inside to $\mathcal{P}$ (Type TI) if and only if $-a^2\leq \lambda_3=-2r=\lambda_4<0$.
		\item[(iii)]  $\mathcal{S}$ and $\mathcal{P}$ are in non-tangent contact (Type C) if and only if $\bar\lambda_3=\lambda_4\in\mathbb{C}-\mathbb{R}$.
		\item[(iv)]  $\mathcal{S}$ is tangent from outside to $\mathcal{P}$ (Type TE) if and only if $0<\lambda_3=2r=\lambda_4\leq a^2$.
		\item[(v)]  $\mathcal{S}$ is exterior to $\mathcal{P}$ (Type E) if and only if $0<\lambda_3<2r<\lambda_4$.
	\end{enumerate}
\end{lemma}

\proof
Non tangent contact between the surfaces is characterized by the condition $|z_c|<r$ and this is just the case when the roots $\lambda_3$ and $\lambda_4$ are complex numbers, $\overline\lambda_3=\lambda_4$ as a direct consequence of \eqref{eq2roots}. This shows $(iii)$.

Due to the ``smallness'' hypothesis, the sole possible point of tangency is $(0,0,0)$ so $\mathcal{S}$ and $\mathcal{P}$ are tangent if and only if $|z_c|=r$. In this cases the discriminant  $\sqrt{z_c^2-r^2}$ in \eqref{eq2roots} vanishes and thus $\lambda_3=\lambda_4$. Furthermore, this root is $\lambda_3=\lambda_4=-2r$ if $z_c=r$ and $\lambda_3=\lambda_4=2r$  if $z_c=-r$. This shows $(ii)$ and $(iv)$.

As the center of the sphere is at the $OZ$-axis, $\mathcal{S}$  is interior to $\mathcal{P}$ if and only if $r<z_c$ whereas $\mathcal{S}$  is exterior to $\mathcal{P}$ if and only if $z_c<-r$. In both cases the term $\sqrt{z_c^2-r^2}$ is positive, so $\lambda_3\neq \lambda_4$. Moreover, if $\mathcal{S} $  is interior to $\mathcal{P}$ then from \eqref{eq2roots} we have $\lambda_3<\lambda_4<0$. Similarly, if $\mathcal{S} $  is exterior to $\mathcal{P}$ then $0<\lambda_3<\lambda_4$.

To show that $\lambda_4\neq -a^2,-b^2$ in the interior case, note that if $z_c= r$, then $-b^2<-a^2\leq -2r=\lambda_4$, but $\lambda_4$ is an strictly increasing function on $z_c$ for $z_c>r$. So   $-b^2<-a^2\leq-2r<\lambda_4$  for  $z_c>r$, this is, when $\mathcal{S}$ is interior to $\mathcal{P}$.  
\qed

Note that the case $2r=-a^2$ is a special one, where interior tangency is characterized by a triple root $-a^2$. This is precisely the case in which the maximum curvature of the vertical parabola in the plane $y=0$ equals the curvature of a maximum circumference of the sphere (see Figure~\ref{figure:parabola1}). A double root $-a^2$ does not correspond to a tangent position, but to the interior case (Type I) in this instance.

\section{Characterization of the relative positions}\label{section:relative positions}
In this section we explore the relation between relative positions  and characteristic roots when $a\neq b$. Again, we assume the ``smallness'' hypothesis and use previous notation. We begin by relating the tangent situation with multiple roots, then we distinguish the interior and exterior cases in terms of the sign of real roots and, finally, we associate non-tangent contact with complex roots.
\subsection{Tangency points and multiple roots}\label{subsect:tangency-multiplicity}
A key point in our global analysis is the relation between tangency and multiplicity of roots. 
Generally speaking, tangency is detected by multiple roots of the characteristic polynomial. However, there exist two multiple roots, namely $-a^2$ and $-b^2$, that can be double and are not associated to a tangent position between the surfaces. The following results express this fact. 
In this section we are going to carefully analyze the relation between the multiple roots and the existence of a tangent point between $\mathcal{S}$  and $\mathcal{P}$. First note that if the quadrics are tangent then there exists a multiple root. The proof of the following two results is  analogous to those given in \cite{BPST_Hyperboloid} (see Lemma~26 and Lemma~25, respectively), so we omit them in the interest of brevity.

\begin{lemma}\label{lemma:tangent-implies-multipleroot}
	If $\mathcal{P}$ and $\mathcal{S}$ are tangent, then there exists a multiple real root of the characteristic polynomial.
\end{lemma}

A partial converse of Lemma~\ref{lemma:tangent-implies-multipleroot} is the following.

\begin{lemma}\label{lemma:multipleroottangent}
	Let $\lambda\notin\{ -a^2, -b^2\}$ be a real root of $f(\lambda)$. If the multiplicity of $\lambda$ is $m\geq 2$, then there exists at least one point where $\mathcal{P}$ and $\mathcal{S}$ are tangent.
\end{lemma}

\begin{remark}\rm
	The roots $-a^2$ and $-b^2$ are indeed special in Lemma~\ref{lemma:multipleroottangent}. The following example shows how $-a^2$ can be a double root and there is no tangency between $\mathcal{S}$  and $\mathcal{P}$:
	\[
	\begin{array}{c}
	\mathcal{P}: \frac{x^2}{1,2^2}+\frac{y^2}{1,5^2}-z=0,\qquad\mathcal{S}:  x^2+(y-0,5)^2+(z-0,712045)^2=0,25^2.\\
	\noalign{\medskip}
	{\bf Roots:}\quad \lambda_1=-3,54808,\quad \lambda_2=\lambda_3= -1,44, \quad\lambda_4=-0,110095.
	\end{array}
	\]
	
	\begin{figure}[h]
		\includegraphics[width=4cm]{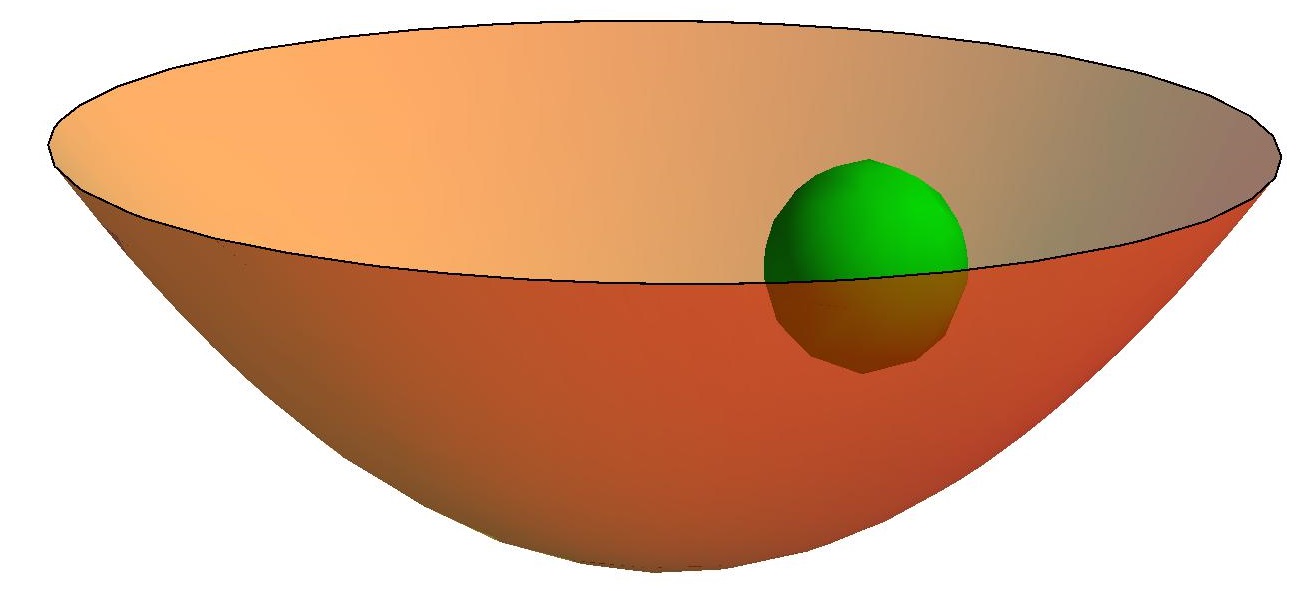}
		\caption{A double root $-a^2$ does not imply that $\mathcal{S}$  and $\mathcal{P}$ are tangent.}
	\end{figure}	
\end{remark}

The following result summarizes the relation between tangency and multiple roots. We leave the proof for the Appendix so that we do not break the flow of the global argument.

\begin{theorem}\rm\label{th:tangency}
	Assume $a\neq b$. Then	$\mathcal{P}$ and $\mathcal{S}$ are tangent if and only if one of the following possibilities holds:
	\begin{enumerate}
		\item    $-a^2$ is a triple root.  
		\item $\lambda\in \mathbb{R}\backslash\{ -a^2, -b^2\}$  is a multiple root.
	\end{enumerate}
\end{theorem}

\subsection{The interior and exterior cases}
We have already proved the characterization of relative positions when  the center of $\mathcal{S}$  is in the $OZ$-axis (in the case $a\neq b$). To extend it to the remaining space we are going to move the center of the sphere along a path from any point to an appropriate point in the $OZ$-axis.
Denote this path by $\alpha(t)$, with $t\in [t_0,t_1]$. Define $f_t(\lambda)=\operatorname{det}(\lambda P+S(t))$ as the characteristic polynomial $f(\lambda)$ for $\mathcal{S}$ centered at $\alpha(t)$. 

\begin{lemma}\label{lemma:complexroots}
Let $p(t)=\alpha_n(t) \lambda^n+\alpha_{n-1}(t) \lambda^{n-1}+\dots +\alpha_1(t) \lambda+ \alpha_0(t)$ be a polynomial of degree $n$ whose coefficients depend continuously on a parameter $t$. Assume $p(t_0)$ has $n$ distinct real roots and that $p(t_1)$ has some complex (non real) root. Then there exists $t_d\in(t_0,t_1)$ so that $p(t_d)$ has a multiple real root.
\end{lemma}
\begin{proof}
Let $t_d=\operatorname{Inf}\{t\in[t_0,t_1]/p(t)\text{ has a non real root}\}$. Note that $t_d\in (t_0,t_1)$. There exist $\varepsilon>0$ small enough so that $p(t_d-\varepsilon)$ has real roots and $p(t_d+\varepsilon)$ has some non real root. For $\varepsilon$ small enough, $p(t)$ factorizes as \[p(t)=(\alpha_2(t) \lambda^2+\alpha_1(t) \lambda+\alpha_0(t)) s(t) \text{ in } [t_d-\varepsilon,t_d+\varepsilon],\] where the second degree polynomial $\alpha_2(t) \lambda^2+\alpha_1(t) \lambda+\alpha_0(t)$ has complex conjugate roots for $t=t_d+\varepsilon$. The functions $\alpha_i$ are continuous functions for $i=0,1,2$, so the discriminant $d(t)=\alpha_1(t)^2-4\alpha_0(t)\alpha_2(t)$ is also a continuous function. Since $d(t_d-\varepsilon)>0$ and $d(t_d+\varepsilon)<0$, it follows that $d(t_d)=0$, so there is a double root for $\alpha_2(t) \lambda^2+\alpha_1(t) \lambda+\alpha_0(t)$ and a root of multiplicity at least two for $p(t_d)$. 
 \end{proof}

\begin{lemma}\label{lemma:non-contacto}
Assume $a\neq b$. Suppose that there is no contact between $\mathcal{P}$ and $\mathcal{S}$. Then $f(\lambda)$ has four real roots, all of them with multiplicity $1$, except $-a^2$ and $-b^2$ that could possibly have multiplicity $2$. Moreover: 
\begin{enumerate}
  \item If $\mathcal{S}$ is interior to $\mathcal{P}$, then all the roots are negative.
  \item If $\mathcal{S}$ is exterior to $\mathcal{P}$, two roots are positive and two are negative and different.
\end{enumerate}
\end{lemma}

\proof
First, let us consider a sphere $\mathcal{S}$ which is interior to the paraboloid. If the center of $\mathcal{S}$ is at the $OZ$-axis, because of Lemma~\ref{centroejez}, the roots are $\lambda_1=-a^2$, $\lambda_2=-b^2$ and $\lambda_3<\lambda_4<0$. If $x_c\neq 0$ and $y_c\neq 0$, we can construct the path $\alpha(t)=((1-t) x_c, (1-t) y_c, z_c)$ with $t\in [0,1]$ (see Figure~\ref{figure:paths}~(a)). This path does not intersect the planes $x=0$ or $y=0$, so $-a^2$ and $-b^2$ are not roots of $f_t(\lambda)$ for $t\in[0,1)$ (see Section~\ref{section:appendix-multiplerootsandtangency}). Since $\mathcal{S}$  is not tangent to $\mathcal{P}$ at any point of the path, there are no multiple roots in virtue of Lemma~\ref{lemma:multipleroottangent}. Moreover, since the roots in $\alpha(1)$ are all real and negative, since $0$ is not a root (see Lemma~\ref{lemma:tecnical-results-1}), and since Lemma~\ref{lemma:complexroots} does not allow complex roots, the roots of $f_0(\lambda)$ are also real and negative. 

Assume $x_c=0$ and $y_c\neq 0$. Then $-a^2$ is a root. Consider the path $\alpha(t)=(0,(1-t)y_c,z_c)$ and note that $\alpha$ is contained in the plane $x=0$, so $-a^2$ is always a root when the center of the sphere moves along $\alpha$. Therefore, the characteristic polynomial decomposes as $f(\lambda)=(\lambda+a^2)g(\lambda)$. We use Lemma~\ref{lemma:tecnical-results-1}, Lemma~\ref{centroejez} and Lemma~\ref{lemma:complexroots} as before, together with Theorem~\ref{th:tangency} to conclude that all the roots are negative in $\alpha(1)$. Note that $g(\lambda)$ does not have any double root, but $-a^2$ could be a root of $g(\lambda)$ ($-a^2$ and $-b^2$ are the only possible double roots without an associated tangency, see Theorem~\ref{th:tangency}). In that case $\mathcal{S}$  is interior to $\mathcal{P}$ and $-a^2$ is a double root of $f(\lambda)$. The argument is similar if $y_c=0$ and $x_c\neq 0$. Thus assertion (1) follows.

To prove assertion (2) we construct  a path so that the sphere moves without  intersecting the paraboloid. For example, for a general $\mathcal{S}$  and $\mathcal{P}$, the sum of the paths $\alpha(t)=(x_c,y_c,z_c-t(|z_c|+2r))$, $t\in[0,1]$, and $\beta(t)=((1-t)x_c,(1-t)y_c,z_c-|z_c|-2r)$,  $t\in[0,1]$, ensures that there is not contact between the surfaces and the center at the end of the path is located in the negative part of the $OZ$-axis (see Figure~\ref{figure:paths}~(a)). Now, a similar argument to that given in assertion (1) applies to prove assertion (2).
\qed

\begin{figure}[h]
	\begin{tabular}{cc}
	\includegraphics[width=5cm]{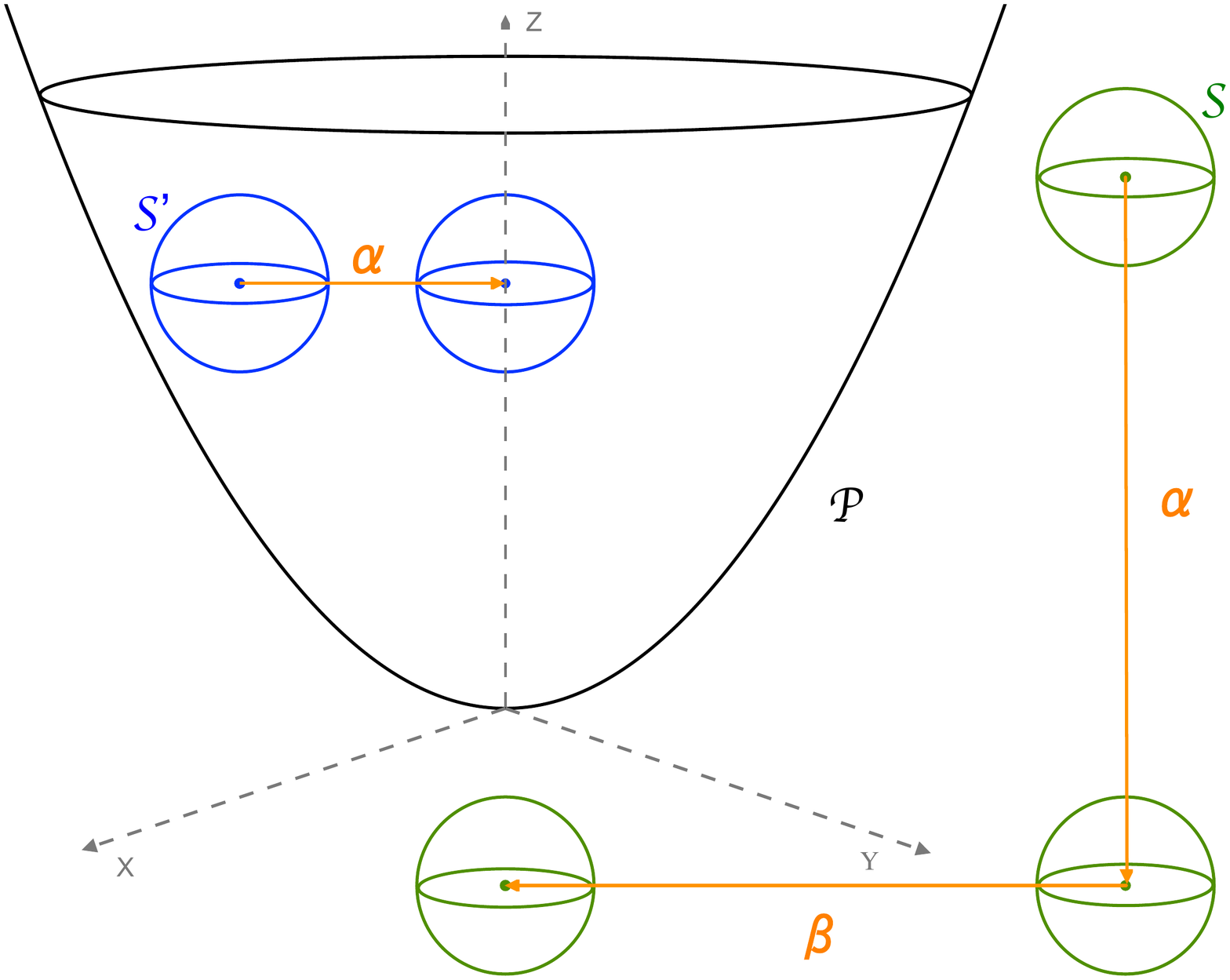} &
	\includegraphics[width=5cm]{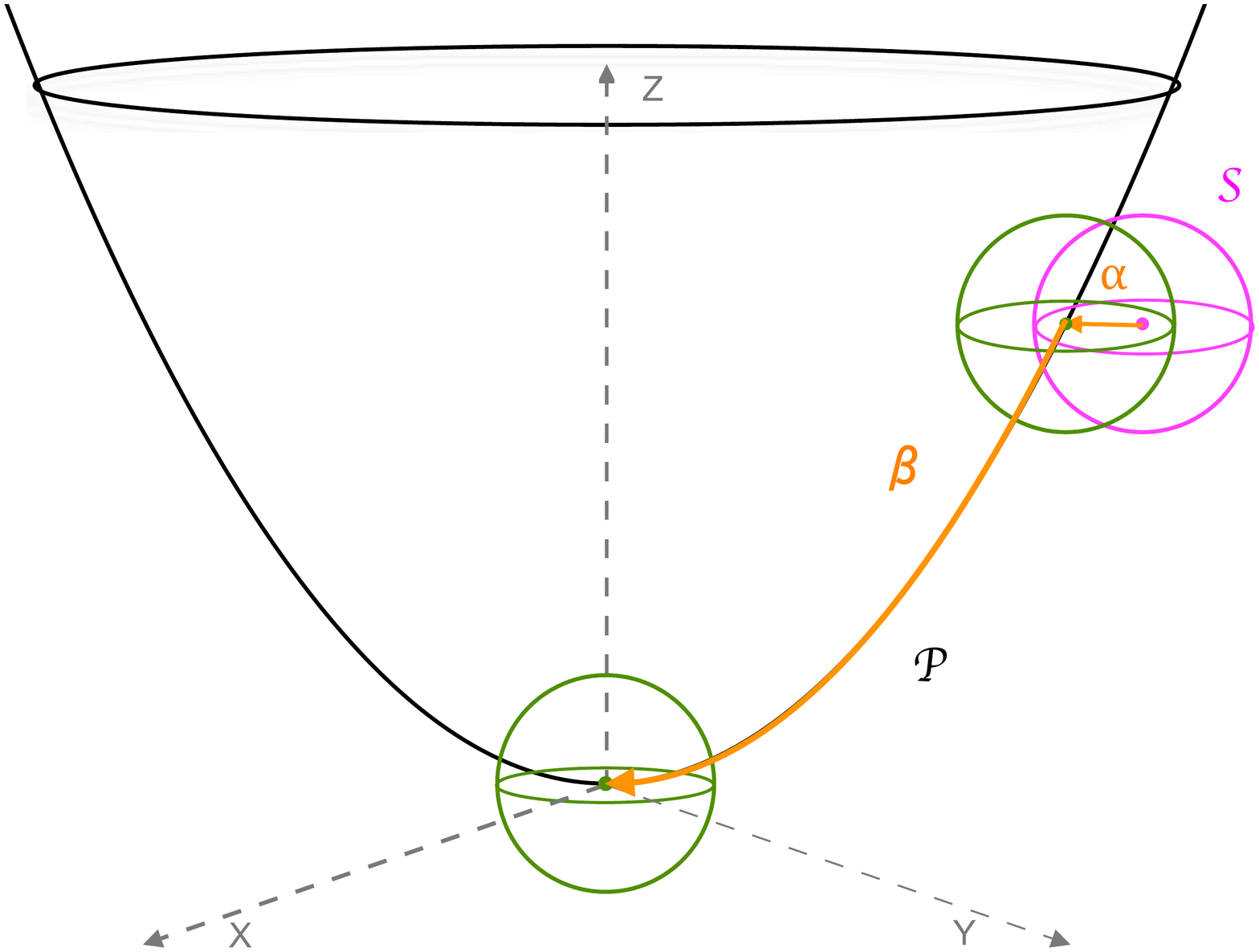} \\
	(a)& (b)
	\end{tabular}
	\caption{Paths translating the center of the sphere from the initial position to the $OZ$-axis with the same relative position all along the path.}\label{figure:paths}
\end{figure}

\subsection{The non-tangent contact case}
\begin{lemma}\label{lemma:contacto-complexas}
Assume $a\neq b$. If $\mathcal{S}$ and $\mathcal{P}$ are in non-tangent contact, then $f(\lambda)$ has a pair of complex conjugate roots and two different real roots.
\end{lemma}

\proof
Assume $\mathcal{S}$ and $\mathcal{P}$ are in non-tangent contact. Then we can build a path from $(x_c,y_c,z_c)$ to $(0,0,0)$ to move the center of $\mathcal{S}$  along it in such a way that all along the path $\mathcal{S}$ and $\mathcal{P}$ are in non-tangent contact. We do it in two steps as follows. Let $(x_1,y_1,z_c)$ be the intersection point between the horizontal half-ray starting at the $OZ$-axis trough $(x_c,y_c,z_c)$. We consider the horizontal path $\alpha(t)=((1-t)x_c+t x_1,(1-t)y_c+t y_1, z_c)$, $t\in [0,1]$,  which joins $(x_c,y_c,z_c)$ and $(x_1,y_1,z_c)$. Now consider the arch of parabola from $(x_1,y_1,z_c)$ to $(0,0,0)$ obtained by intersecting $\mathcal{P}$ with the vertical plane which contains these two given points and parametrized it as the curve $\beta$. The sum of $\alpha$ and $\beta$ provides a path from $(x_c,y_c,z_c)$ to $(0,0,0)$ (see Figure~\ref{figure:paths}~(b)). 

By Lemma~\ref{centroejez}, the characteristic polynomial has a pair of complex conjugate roots when the center of $\mathcal{S}$  is at $(0,0,0)$. Assume $x_c\neq 0 \neq y_c$, so $-a^2$ and $-b^2$ are not roots. Since there are no tangent points along the path we have built, there are no roots with multiplicity greater than $1$ (see Lemma~\ref{lemma:multipleroottangent}), and hence there are a pair of complex roots when the center of $\mathcal{S}$  is at $(x_c,y_c,z_c)$ as a consequence of Lemma~\ref{lemma:complexroots}. If $x_c=0$, then $-a^2$ is a root (see Lemma~\ref{lemma:a2-multipleroot}) and the images of $\alpha$ and $\beta$ belong to the plane $x_c=0$. Hence we decompose $f(\lambda)=(\lambda+a^2)h(\lambda)$   and argue as before with the polynomial $h$ to conclude the result.   
\endproof

\section{The proof of Theorem~\ref{th:th1}}\label{section:theproofofth}

\noindent{\it Proof of Theorem~\ref{th:th1}.} This section is devoted to prove the classification result of Table~\ref{table:main}. We analyze the cases $a\neq b$ and $a=b$ separately given their different behavior.

\subsection{The strictly elliptic case: $a\neq b$.}
We begin with a general ellipsoid and an elliptic paraboloid in a certain relative position. 
Since the characteristic roots of $\mathcal{E}$ and $\mathcal{P}$ are invariant under affine transformations, we apply the corresponding homotheties and rigid moves to the space so that the ellipsoid becomes a sphere and the elliptic paraboloid is in standard form (see Section~\ref{section:characteristic-polynomial}).  We assume first $a\neq b$ and use previous results as follows.

As a direct consequence of Lemma~\ref{lemma:non-contacto}, if $\mathcal{S}$  is interior to $\mathcal{P}$ then there are four real roots and either all of them are different or some of $-a^2$ and $-b^2$ are double roots (see Section~\ref{section:appendix-multiplerootsandtangency} for the $-a^2$ and $-b^2$ exceptional cases). Also, if $\mathcal{S}$  is exterior to $\mathcal{P}$ then there are two real positive and two real negative simple roots.

Lemma~\ref{lemma:contacto-complexas} shows that there are two complex conjugates roots in the non-tangent contact case. Note that the argument in Lemma~\ref{lemma:contacto-complexas} also shows that the other two roots are real.

Finally, the tangent cases are associated with a multiple root by Theorem~\ref{th:tangency}. Moreover, tangent relative positions can be obtained by moving a sphere continuously from an interior or an exterior position. Since the roots are continuous functions on the coefficients of the characteristic polynomial, the interior tangency results in four real roots, at least one of which has multiplicity strictly greater than $1$, and the exterior tangency results in two negative roots and one double root which is positive.

\subsection{The circular paraboloid case: $a=b$}\label{section:appendix-a=b}
To finish the proof of Theorem~\ref{th:th1} it remains to consider the case $a=b$. 
If $a=b$ the paraboloid is circular and is an exceptional case in the analysis developed in Sections~\ref{section:characteristic-polynomial} and \ref{section:relative positions}. Situations like this, where the quadric is a surface of revolution, were considered previously in the literature (see \cite{BPST_Hyperboloid}). When we work with a circular paraboloid and a sphere, $-a^2$ is always a root of the characteristic polynomial. Furthermore, due to the rotational symmetry, the problem can be reduced in one dimension by considering the intersection of the plane that contains the axis of the paraboloid and the center of the sphere. Hence, $f(\lambda)=(\lambda+a^2)h(\lambda)$, so one needs to study the third degree polynomial $h$ and the relative positions between a circumference and a parabola.

However,  the circular case is obtained using continuity of the roots of the characteristic polynomial from a paraboloid with  $a\neq b$. Since we already have the classification of the relative position if $a\neq b$, we opt for a direct approach based on the continuity of the roots of $f$ (see \cite{kato} for a complete reference). Thus, consider an $\epsilon$-paraboloid of the form  
\[
\frac{x^2}{a^2}+\frac{y^2}{(a+\epsilon)^2}-z=0,  \text{ with } \epsilon>0.
\]
For a given relative position of the circular paraboloid, it can be obtained as a limit of an $\epsilon$-paraboloid when $\epsilon\to 0$. Thus, since $0$ is never a root by Lemma~\ref{lemma:tecnical-results-1}, a position where $\mathcal{S}$  is interior to $\mathcal{P}$ and $\mathcal{S}$ is tangent to $\mathcal{P}$ from inside has roots $\lambda_3\leq \lambda_4<0$ for any $\epsilon$-paraboloid, so in the limit, when $\epsilon \to 0$, also has $\lambda_3\leq \lambda_4<0$. The fact that a position of contact between $\mathcal{S}$  and $\mathcal{P}$ provides complex roots for $f$ is obtained by redoing the argument of Lemma~\ref{lemma:contacto-complexas} for the circular paraboloid, taking into account that $-a^2$ is always a root and working with a third degree polynomial (see also \cite{BPST_Hyperboloid} for an analogous situation with a circular hyperboloid). For any $\epsilon$-paraboloid, that $\mathcal{S}$  is tangent from outside to $\mathcal{P}$ is characterized by a positive double root $\lambda_3=\lambda_4>0$, so when  $\epsilon \to 0$ this will happen too, because $0$ is not a root. Finally, if $\mathcal{S}$  is exterior to $\mathcal{P}$, then $0<\lambda_3<\lambda_4$ for any  $\epsilon$-paraboloid, so in the limit we have $0<\lambda_3\leq\lambda_4$ because $0$ is not a root. Now, Lemma~\ref{lemma:multipleroottangent} excludes the case $\lambda_3=\lambda_4$ which corresponds to tangent contact also if $a=b$. In summary, the classification given in Theorem~\ref{th:th1} is also valid when $a=b$. This completes the proof of Theorem~\ref{th:th1}.
\endproof

\section{Conclusions}\label{section:conclusions}
We have shown that the roots of a characteristic polynomial of order four are suitable to detect the relative position between an ellipsoid and an elliptic paraboloid if the ellipsoid is small in comparison with the elliptic paraboloid. Table~\ref{table:main} summarizes the five possible relative positions in terms of the roots. This classification sets the theoretical framework to develop efficient algorithms based on the analysis of the coefficients of the characteristic polynomial as in Table~\ref{table:coefficients}. Thus, contact detection between the two quadrics is simple enough to be applied in a continuous time-varying positional context (see those algorithms introduced in \cite{alberich2017,etayo2006} in an analogous context and references therein).

The exterior case is especially interesting, as the relative positions are directly detected as follows: the surfaces are not in contact if there are two different positive roots, they are tangent if there is a double positive root and they are in non-tangent contact if there are complex conjugate roots. This classification agrees with that given for two ellipsoids or an sphere and a circular hyperboloid (see \cite{BPST_Hyperboloid,wank-wang-kim}).

\section{Appendix}

\subsection{The ``smallness'' condition in more detail}\label{section:appendix-smallness}
The ``smallness'' hypothesis can be interpreted in terms of the principal curvatures of the surfaces. The fact that the maximum principal curvature of the paraboloid is less than or equal to the minimum curvature of the ellipsoid guarantees the ``smallness'' condition. Thus, for a paraboloid $\mathcal{P}$ given by equation \eqref{eq:paraboloid}, we intersect $\mathcal{P}$ with the vertical plane $y=0$ to obtain the parabola $z=\frac{x^2}{a^2}$. Since we are assuming $a\leq b$, this parabola is the one with greatest curvature. We parametrize the parabola by  $\alpha(t)=(t,{t^2}/{a^2})$ and compute the curvature using the expression $\kappa(t)=\frac{\|\alpha'(t)\times \alpha''(t)\|}{\|\alpha'(t)\|^3}$ (see, for example, \cite{docarmo}) to obtain
\[
\kappa(t)=\frac{2 a^4}{\left(\sqrt{a^4+4t^2}\right)^3}.
\]
The curvature attains its maximum at $t=0$, this is, the vertex of the parabola, and $\kappa(0)=\frac{2}{a^2}$.

In general, if we consider an ellipsoid, we shall compare the minimal principal curvature of the ellipsoid with the value $2/a^2$. Thus, for example, for an ellipsoid given in standard form:
\[
\frac{x^2}{c_1^2}+\frac{y^2}{c_2^2}+\frac{z^2}{c_3^2}=1 \qquad \text{ with } c_1\leq c_2 \leq c_3,
\]
the minimum principal curvature is attained at the co-vertices of the ellipse at the $XZ$-plane. The value of the curvature at that point is $\frac{c_1}{c_3^2}$. Therefore, the ``smallness'' hypothesis reads $\frac{c_1}{c_3^2}\geq  \frac{2}{a^2}$ in terms of the parameters of the quadrics.

\begin{figure}[h]\label{figure:contraexemplo}
	\includegraphics[width=3cm]{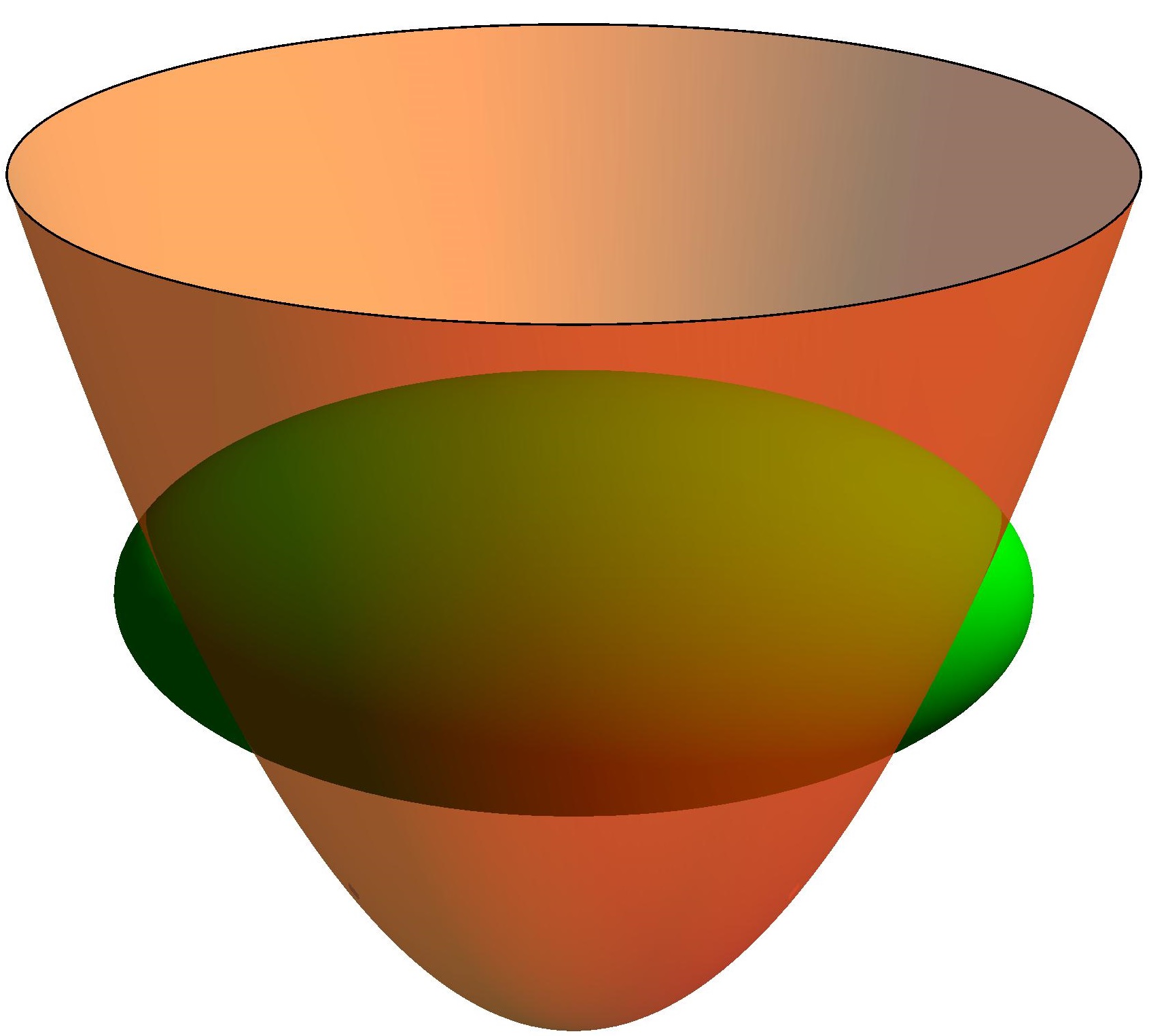}
	
	\caption{The paraboloid $x^2+\frac{y^2}{16}=z$ is intersected by the ellipsoid $\frac{x^2}{4}+y^2+4 z^2=1$ in two connected components, since the minimum principal curvature of the ellipsoid is $\frac{1}{8}$ and the maximum principal curvature of the paraboloid is $2$, so the ``smallness'' condition is not satisfied.}
\end{figure}


\subsection{The proof of Theorem~\ref{th:tangency}}\label{section:appendix-multiplerootsandtangency}

We have seen in Section~\ref{section:relative positions} that the roots $-a^2$ and $-b^2$ are 
 exceptional. We analyze their possible multiplicities as follows.
\begin{lemma}\label{lemma:a2-multipleroot}
	Assume $a\neq b$. Then:
	\begin{itemize}
		\item[(i)] $-a^2$ is a root if and only if $x_c=0$.
		\item[(ii)] If $-a^2$ is a root, then it has multiplicity at least $2$ if and only if 
		 $r^2=a^2z_c-\dfrac{a^2}{b^2-a^2}y_c^2-\dfrac{a^4}{4}$.
		 \item[(iii)] $-a^2$ has multiplicity $3$ if and only if $x_c=y_c=0$ and $z_c=r=\frac{a^2}2$.
	\end{itemize}
\end{lemma}
\proof
First, substitute using expression \eqref{eq:characteristic-polynomial}: $f(-a^2) =x_c^2 (-b^2+ a^2)/b^2$ to check that $f(-a^2)=0$ if and only if $x_c=0$. This proves assertion $(i)$.

Assume $x_c=0$. Then $f$ decomposes as $f(\lambda)=(a^{-2}\lambda+1)g(\lambda)$ where
$g(\lambda)=-(b^{-2}\lambda+1) \left(r^2+\lambda^2/4+z_c\lambda\right)+b^{-2}\lambda y_c^2$.
Thus $g(-a^2) = 0$ if and only if $r^2=a^2 z_c-a^4/4-a^2y_c^2/(b^2-a^2)$, so assertion $(ii)$ follows.

Assume $x_c=0$ and $r^2=a^2z_c-{a^2}/({b^2-a^2})y_c^2-{a^4}/{4}$. Then $f$ decomposes as $f(\lambda)=(a^{-2}\lambda+1)^2h(\lambda)$ where 
\[
h(\lambda)=\lambda^2 +(4z_c+b^2-a^2)\lambda+4b^2z_c-a^2b^2-\frac{4y_c^2b^2}{b^2-a^2}.
\]
We have that, $h(-a^2) = 0$  if and only if $z_c= a^2/2+y_c^2 b^2 /(b^2-a^2)^2$. Substituting $z_c$ in  $r^2=a^2z_c-{a^2}/(b^2-a^2)y_c^2-{a^4}/{4}$ we obtain $r^2=a^4\left(y_c^2/(b^2-a^2)^2+ 1/4\right)$. Since $2r\leq a^2$ by the ``smallness'' condition, we conclude that $y_c^2/(b^2-a^2)^2 \leq 0$, so necessarily $y_c=0$. Hence $h(-a^2) = 0$ implies that $y_c=0$ and $2r=a^2=z_c$ as in assertion $(iii)$. The converse is immediate.
\qed

Note that the condition for $-a^2$ to be a double root can be rewritten as $z_c=\dfrac{y_c^2}{b^2-a^2}+\dfrac{r^2}{a^2}+\dfrac{a^2}{4}$. This expression evidences the fact that if $-a^2$ is a double root, then the center of $\mathcal{S}$, $(0,y_c,z_c)$, is interior to $\mathcal{P}$.

Analogous assertions to $(i)$ and $(ii)$ in Lemma~\ref{lemma:a2-multipleroot} follow for the root $-b^2$. However, $-b^2$ cannot be a triple root if $a\neq b$, since for that to happen it is necessary that $r=b^2/2$ and that contradicts the ``smallness'' assumption.

{\it Proof of  Theorem~\ref{th:tangency}}. We need to show that a root $-a^2$ with multiplicity equal to $2$ is not associated with tangency between $\mathcal{S}$  and $\mathcal{P}$. Then Theorem~\ref{th:tangency} will follow from Lemma~\ref{lemma:tangent-implies-multipleroot}, Lemma~\ref{lemma:multipleroottangent} and Lemma~\ref{lemma:a2-multipleroot} taking into account that $-b^2$ is never a triple root. 

If $-a^2$ is a root, we have a tangent point associated to it at $X_0$ if and only if $X_0$ is an  eigenvector (meaning it is a solution of $(-a^2 P+S)X=0$) and verify either $X_0^t\mathcal{S} X_0=0$ (these imply $X_0^t\mathcal{P} X_0=0$).
	
Assume $-a^2$ is a root of multiplicity $2$. Then $x_c=0$ and  $r^2=a^2z_c-{a^2}/(b^2-a^2)y_c^2-{a^4}/{4}$.
For $X=(x,y,z,1)^t$, we solve $(-a^2 P+S)X=0$ where
	\[
	-a^2 P+S=\left(   \begin{array}{cccc}
	0&0&0&0 \\
	0&-a^2 b^{-2}+1&0&-y_c \\
	0&0&1&-(z_c -a^2 /2)\\
	0&-y_c&-(z_c -a^2 /2)&-r^2 +y_c^2+z_c^2
	\end{array} \right ),
	\]
to obtain  $$(0,\frac{(b^2-a^2)y}{ b^2}- y_c,     -z_c+z+\frac{a^2}2,    z_c^2+z(\frac{a^2}2-z_c)+y_c^2-yy_c-r^2    )=0.$$ 
These equations have a solution of the form  
\begin{equation}\label{eq:form-tangent-vector}
x\in {\mathbb R}, \,  y=\frac{b^2y_c}{b^2-a^2}, \, z=z_c-\frac{a^2}2
\end{equation}  
where  $ z_c^2+(z_c-a^2/2)(a^2/2-z_c)+y_c^2-(b^2y_c/(b^2-a^2))y_c-r^2 =0$. Substitute $y$ and $z$ in the later condition to reduce it to
 $ z_ca^2 -y_c^2 a^2/(b^2-a^2) -a^4/4-r^2 =0$, which is satisfied because $-a^2$ is a double root, so this later condition impose no new restrictions.
	
Now, since $(x,y,z)$ is a tangent point we have that $r^2=x^2+(y-y_c)^2+(z-z_c)^2$. We use $z=z_c-a^2/2$ from \eqref{eq:form-tangent-vector} to see  that 
$$r^2=x^2+(y-y_c)^2+\frac{a^4}4\geq \frac{a^4}4$$
and due to the ``smallness'' assumption ($2r\leq a^2$) we have that $r=a^2/2$. Moreover, from \eqref{eq:form-tangent-vector} we conclude $x=0$ and $y=y_c=0$, so $-a^2$ is indeed a triple root by Lemma~\ref{lemma:a2-multipleroot}.
	A similar argument exclude the possibility of having a tangent point associated to the root $-b^2$ with multiplicity $2$.
	\qed


\end{document}